\documentclass[runningheads]{llncs}
\usepackage{graphicx}

\usepackage{times}
\usepackage{epsfig}
\usepackage{amsmath}
\usepackage{amssymb}
\usepackage{enumitem}

\usepackage{bm}
\usepackage{upgreek}
\usepackage{algorithmic}
\usepackage[vlined,ruled]{algorithm2e}
\usepackage{soul,xcolor}
\usepackage[bottom]{footmisc}
\usepackage{mathtools}
\usepackage[breaklinks=true,colorlinks=true,citecolor=blue]{hyperref}

\spnewtheorem{rmk}{Remark}[section]{}{\itshape}
\DeclareMathOperator*{\argmin}{argmin}

\begin{document}
\title{
Controlling Meshes via Curvature: Spin Transformations for Pose-Invariant Shape Processing}
\titlerunning{Spin Transformations for Pose-Invariant Shape Processing}
%
\author{Lo\"ic Le Folgoc\inst{1} 
\and
Daniel C. Castro\inst{1} \and
Jeremy Tan\inst{1} \and
Bishesh Khanal\inst{2} \and \\
Konstantinos Kamnitsas\inst{1} \and
Ian Walker\inst{1} \and
Amir Alansary\inst{1} \and
Ben Glocker\inst{1}}
\authorrunning{L. Le Folgoc et al.}
%
\institute{BioMedIA, Imperial College London, United Kingdom \\
\and King's College London, United Kingdom\\
\email{l.le-folgoc@imperial.ac.uk}
}
\maketitle

\newcommand{\bigcomment}[2]{
  \begin{center}
  \fbox{\begin{minipage}{\linewidth}{\bf #1:} {\rm #2}\end{minipage}}
  \end{center}
}
\newcommand{\AN}[1]{\bigcomment{AN}{#1}}
\newcommand{\LLF}[1]{\bigcomment{LLF}{#1}}

\newcommand{\tool}{\textsc{CLDF}\xspace}
\newcommand{\todo}[1]{{\color{red} \textbf{TODO:}\xspace #1}}

\newcommand{\R}{\mathbb{R}}
\newcommand{\bbC}{\mathbb{C}}
\newcommand{\bbH}{\mathbb{H}}
\newcommand{\bbi}{\boldsymbol{\mathrm{i}}}
\newcommand{\bbj}{\boldsymbol{\mathrm{j}}}
\newcommand{\bbk}{\boldsymbol{\mathrm{k}}}
\newcommand{\calB}{\mathcal{B}}
\newcommand{\calC}{\mathcal{C}}
\newcommand{\calD}{\mathcal{D}}
\newcommand{\calE}{\mathcal{E}}
\newcommand{\calF}{\mathcal{F}}
\newcommand{\calG}{\mathcal{G}}
\newcommand{\calH}{\mathcal{H}}
\newcommand{\calI}{\mathcal{I}}
\newcommand{\calL}{\mathcal{L}}
\newcommand{\calM}{\mathcal{M}}
\newcommand{\calN}{\mathcal{N}}
\newcommand{\calO}{\mathcal{O}}
\newcommand{\calR}{\mathcal{R}}
\newcommand{\calS}{\mathcal{S}}
\newcommand{\calT}{\mathcal{T}}
\newcommand{\calV}{\mathcal{V}}
\newcommand{\calX}{\mathcal{X}}
\newcommand{\calW}{\mathcal{W}}
\newcommand{\rmA}{\mathrm{A}}
\newcommand{\rmC}{\mathrm{C}}
\newcommand{\rmD}{\mathrm{D}}
\newcommand{\rmE}{\mathrm{E}}
\newcommand{\rmF}{\mathrm{F}}
\newcommand{\rmH}{\mathrm{H}}
\newcommand{\rmI}{\mathrm{I}}
\newcommand{\rmJ}{\mathrm{J}}
\newcommand{\rmL}{\mathrm{L}}
\newcommand{\rmM}{\mathrm{M}}
\newcommand{\rmQ}{\mathrm{Q}}
\newcommand{\rmR}{\mathrm{R}}
\newcommand{\rmU}{\mathrm{U}}
\newcommand{\rme}{\mathrm{e}}
\newcommand{\rmf}{\mathrm{f}}
\newcommand{\rmh}{\mathrm{h}}
\newcommand{\rmn}{\mathrm{n}}
\newcommand{\rmp}{\mathrm{p}}
\newcommand{\rmq}{\mathrm{q}}
\newcommand{\rmu}{\mathrm{u}}
\newcommand{\rmv}{\mathrm{v}}
\newcommand{\rmx}{\mathrm{x}}
\newcommand{\rmy}{\mathrm{y}}
\newcommand{\bmc}{\bm{c}}
\newcommand{\bme}{\bm{e}}
\newcommand{\bmf}{\bm{f}}
\newcommand{\bmh}{\bm{h}}
\newcommand{\bms}{\bm{s}}
\newcommand{\bmq}{\bm{q}}
\newcommand{\bmt}{\bm{t}}
\newcommand{\bmr}{\bm{r}}
\newcommand{\bmu}{\bm{u}}
\newcommand{\bmw}{\bm{w}}
\newcommand{\bmy}{\bm{y}}
\newcommand{\bmz}{\bm{z}}
\newcommand{\brmf}{\boldsymbol{\rmf}}
\newcommand{\brmh}{\boldsymbol{\rmh}}
\newcommand{\brmp}{\boldsymbol{\rmp}}
\newcommand{\brmx}{\boldsymbol{\rmx}}
\newcommand{\brmy}{\boldsymbol{\rmy}}
\newcommand{\brmA}{\boldsymbol{\rmA}}
\newcommand{\brmC}{\boldsymbol{\rmC}}
\newcommand{\brmD}{\boldsymbol{\rmD}}
\newcommand{\brmE}{\boldsymbol{\rmE}}
\newcommand{\brmH}{\boldsymbol{\rmH}}
\newcommand{\brmI}{\boldsymbol{\rmI}}
\newcommand{\brmL}{\boldsymbol{\rmL}}
\newcommand{\brmQ}{\boldsymbol{\rmQ}}
\newcommand{\brmR}{\boldsymbol{\rmR}}
\newcommand{\dis}{\displaystyle}
\newcommand{\T}{{\mkern-1.5mu\mathsf{T}}}
\newcommand{\Id}{\text{Id}}
\newcommand{\tr}{\text{tr}}
\newcommand{\etal}{{\it et al }}
\newcommand{\eg}{\textit{e.g.}\xspace}
\newcommand{\ie}{\textit{i.e.}\xspace}
\newcommand{\uline}[1]{\underline{#1}}
\newcommand{\duline}[1]{\underline{\underline{#1}}}
\newcommand{\eq}{\!=\!}
\newcommand{\mrho}{{\!\rho}}


\def\tool{{\textsc{Eugene}}}

\begin{abstract}
We investigate discrete spin transformations, 
a geometric framework to manipulate surface meshes 
by controlling mean curvature.  Applications include surface fairing -- flowing a mesh onto say, a reference sphere -- and mesh extrusion -- \eg, rebuilding a complex shape from a reference sphere and curvature specification. 
Because they operate in curvature space, these operations can be conducted very stably across large deformations with no need for remeshing. Spin transformations add to the algorithmic toolbox for pose-invariant shape analysis. Mathematically speaking, mean curvature is a shape invariant and in general fully characterizes closed shapes (together with the metric). Computationally speaking, spin transformations make that relationship explicit. 
Our work expands on a \textit{discrete} formulation of spin transformations. Like their smooth counterpart, discrete spin transformations are naturally close to conformal (angle-preserving). This quasi-conformality can nevertheless be relaxed to satisfy the desired trade-off between area distortion and angle preservation. We derive such constraints and propose a formulation in which they can be efficiently incorporated. The approach is showcased on subcortical structures. 
\end{abstract}
\section{Introduction}

Generative shape models are tremendously useful in computational anatomy (shape representation, population analysis), medical imaging and computer vision (segmentation, tracking), computer graphics and beyond. Most approaches to statistical shape analysis fundamentally rely on registration,
from landmark based representations and active shape models~\cite{cootes1995active,belongie2002shape,myronenko2010point},
to medial representations~\cite{joshi2002multiscale} and Principal Geodesic Analysis~\cite{fletcher2004principal},
to deformable registration and diffeomorphometry~\cite{durrleman2014morphometry,zhang2015bayesian}.
Registration is known to be a source of bias in shape analysis, but is often a necessary `evil' because input data does not come pre-aligned in a common reference frame (or \textit{pose}). In contrast, the shape information of interest is often invariant to the object pose. Our main motivation is to investigate geometric tools that can open the way to learned, \textit{pose-invariant} generative shape models (specifically, curves and $3$D surfaces). The key insight is that mean curvature is pose-invariant and generally characterizes the shape \textit{losslessly}. This work investigates spin transformations as the algorithmic tool to computationally implement this insight. The cornerstone of the framework lies in a gracefully simple equation that relates a spin transformation $\phi\colon \calF\rightarrow\bbH$ (one quaternion per face in the mesh) 
to a change $\mu\colon\calF\rightarrow\R$ in the mean curvature via a first-order differential operator $D_e$: 
\begin{equation}
D_e \phi = \mu \phi \,.
\end{equation}

Typically, a desired change of curvature is specified via $\mu$, yielding a transformation $\phi$ from which a new shape can be constructed. 
The present work demonstrates this concept and shows its applicability to manipulate (flow and extrude) closed shapes 
in a stable manner across large deformations. 
Section~\ref{sec: geometric setting} reviews the discrete geometric setting, \ie (i)~the geometric objects to which the framework applies, (ii)~discrete mean curvature, (iii)~background on quaternions as similarity transformations in $\R^3$. Section~\ref{sec: spin transformations} introduces discrete spin transformations. Within the framework of spin transformations, the task of flowing a mesh onto a reference shape and that of extruding a shape back from the reference are highly symmetric: both rely on the ability to compute a transformation based on prescribed curvature and area changes. Section~\ref{sec: algorithms} gives an overview of the proposed procedure. Section~\ref{sec: applications} discusses applications and results.

\begin{figure}[t]
\includegraphics[width=\textwidth]{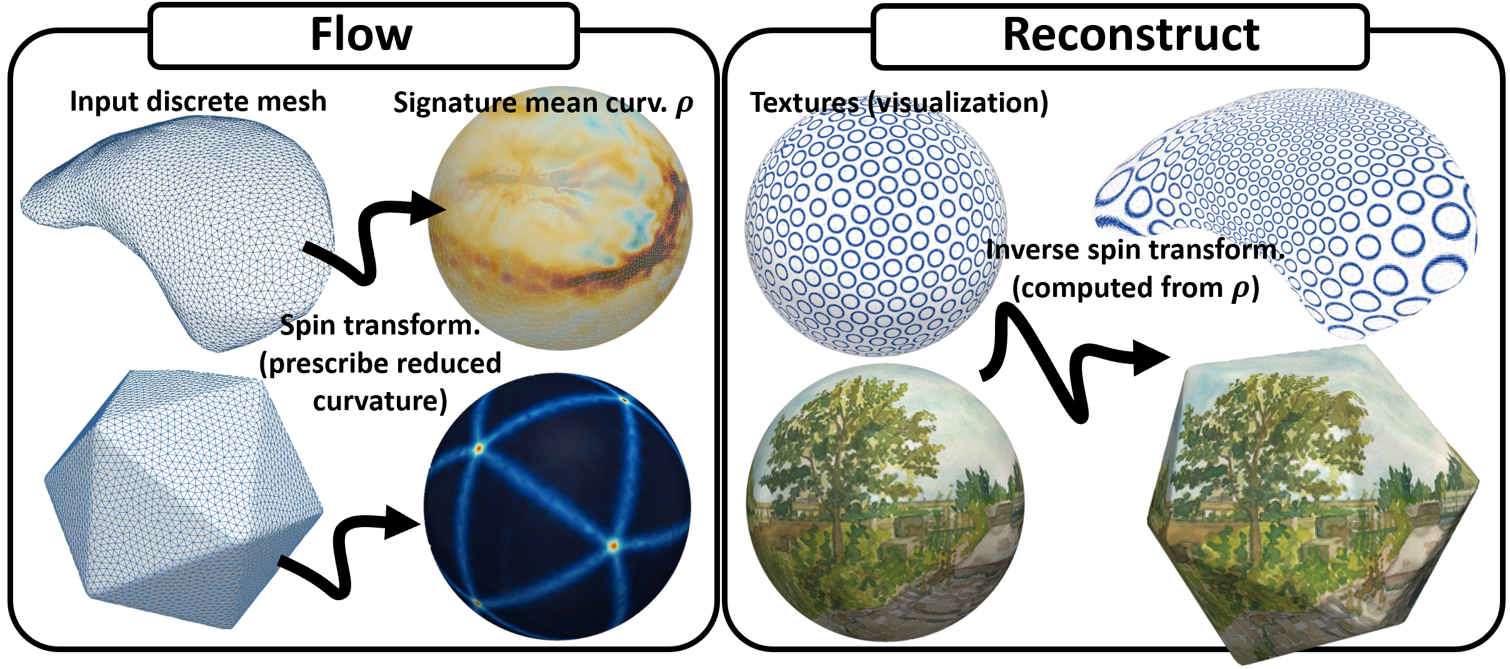}
\caption{Discrete spin transformations allow for controlling meshes via the mean curvature invariant. (a) Input face edge-constraint nets are flowed to a reference shape in the homotopy class (\eg the unit sphere $\calS^2$). Information required to recompute the original shape up to pose and scale is summarized within a scalar field $\rho$. (b) The inverse spin transformation is retrieved. Texture coordinates mapped onto the reference sphere are pushed forward with the extruded mesh. Note the preservation of texture, from which deformations are seen to be quasi-conformal. Top row: putamen. Bottom: icosahedron.} \label{fig: leading figure}
\end{figure}

\section{Related work}

\subsubsection{Pose-invariant shape analysis.} Spectral shape descriptors~\cite{reuter2006laplace,raviv2010volumetric},
built from the spectrum and eigenfunctions of the Laplace(--Beltrami) operator, have achieved popularity in this context,
spanning a variety of applications \eg, object retrieval~\cite{bronstein2011shape}, shape dissimilarity quantification~\cite{konukoglu2013wesd}, analysis of anatomical structures~\cite{niethammer2007global,germanaud2012larger,wachinger2015brainprint}, transfer of structural and functional data~\cite{ovsjanikov2012functional,lombaert2015brain}.
Spectral representations pose two challenges: firstly, going back from the spectral descriptor to the corresponding shape is difficult; and secondly, they tend to discard fine-grained, local information in favor of global shape properties and symmetries. To supplement the intrinsic Laplace--Beltrami operator,  the \textit{extrinsic} Dirac operator~\cite{liu2017dirac}, which carries more information about the shape immersion, has recently been investigated for shape analysis. Geometric deep learning~\cite{bronstein2017geometric} provides the toolset to analyze \textit{functions} over a \textit{fixed} graph. It remains unclear how to analyze \textit{graphs} themselves. The present work contributes with a lossless and invertible mechanism for turning a mesh into a function (the curvature) over a reference template (say, a sphere).

\subsubsection{Shape flows, large deformations, conformal maps.} Mean curvature flow is the archetypal algorithm for fairing, in part due to its simplicity and intuitive appeal. However mesh quality tends to degrade quickly throughout the flow, requiring tedious monitoring and remeshing to reduce artefacts and prevent singularities~\cite{kazhdan2012can}. Furthermore it is not suitable for mesh extrusion. Conformal maps are often perceived as the gold standard in such contexts, and spin transformations originate from this perspective~\cite{crane2011spin,crane2013robust}. Several discretized and discrete quasi-conformal frameworks have been proposed (\eg,~\cite{luo2004combinatorial,lam2018infinitesimal}) on top of an incredibly rich body of theoretical work. Conformal maps have found a natural application in the context of brain mapping~\cite{hurdal2009discrete,gu2004genus} by mapping the cortical surface to a reference domain. Rather than strictly on conformality, our focus here is on a parametrization of large deformations that (1) works from the shape invariant mean curvature (2) allows to efficiently flow between any shape and a reference. It is more generally related in spirit to large diffeomorphic frameworks~\cite{vaillant2005surface,lorenzi2011schild} that can flow a shape from a template and (pose-equivariant) vector field. Our work expands on the framework of \textit{discrete} spin transformations as introduced by Ye et al.~\cite{ye2018unified}. One of the appeals of a discrete framework is to bypass \textit{discretization} errors by design and to offer a consistent definition of discrete geometric concepts such as curvature. We introduce the framework to the community and contribute~(i) with an optimization strategy that gives finer-grained control over deformations;~(ii) by deriving constraints within this formulation for integrability on general topologies, and area preservation;~(iii) by exploring its potency for mesh extrusion.

\section{Discrete Geometric Setting}
\label{sec: geometric setting}

\paragraph{Face edge-constraint nets.} Our work focuses on the case of closed compact orientable surfaces in $\R^3$ and follows the discrete geometric setting introduced by Ye et al.~\cite{ye2018unified}. Surfaces are discretized as face edge-constraint nets, generic constructs that encompass but are not restricted to standard triangulated meshes. Let $\calG=(\calV,\calF,\calE)$ denote the net combinatorics, resp. its vertices, faces and edges. Adjacent faces meet along a single edge. Edges are shared by exactly two adjacent faces. Faces can be arbitrary polygons (such as with simplex meshes~\cite{delingette1999general}). In addition, let each face be assigned a unit normal $n$, such that for any two adjacent faces $i$ and $j$ joined along edge $e_{ij}$ (${(i,j) \in \calE}$), the normals satisfy the looser condition $n_i+n_j \perp e_{ij}$. $\calX=(\calG,n)$ is called a face edge-constraint net. For instance standard triangulations with normals orthogonal to faces are face edge-constraint nets.

\begin{figure}[t]
\includegraphics[width=\textwidth]{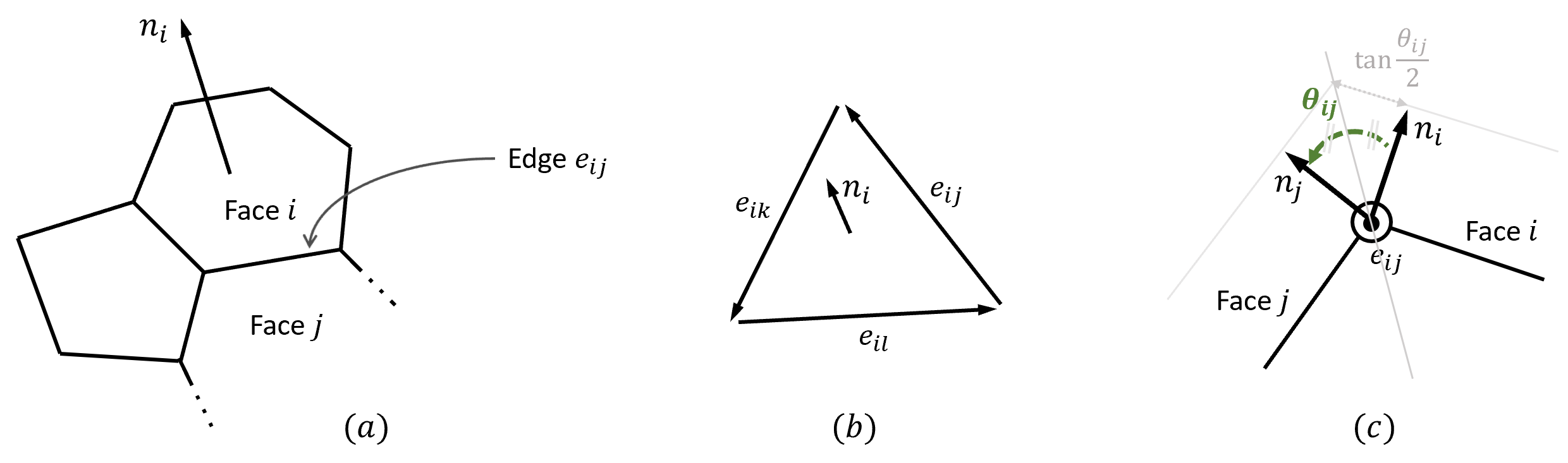}
\caption{Face edge-constraint nets: (a) faces are general polygons; (b) face edges are oriented (counter-clockwise); (c) $\theta_{ij}$ is the bending angle, positive if the edge is convex; the edge integrated mean curvature $H_{ij}\eq \vert e_{ij}\vert \tan(\theta_{ij}/2)$ is the signed created area for face $i$ above $e_{ij}$ when faces $i$ and $j$ are offset by a unit length $1$ in the direction of their normals.} \label{fig: face_edge-constraint_net}
\end{figure}

\paragraph{Discrete mean curvature.} Let $\calX$ be a net. $\calX$ is orientable and, without loss of generality, directed edges $e_{ij}$ are traversed in the direction towards which they point when cycling over vertices of face $i$. 
This lets us orient the dihedral angle $\theta_{ij}$ between planes $P_i\!\triangleq\!\mathrm{span}\{n_i,e_{ij}\}$ and $P_j\!\triangleq\!\mathrm{span}\{n_j,e_{ij}\}$. For standard triangulations, $\theta_{ij}$ is just the bending angle between faces. The \textit{integrated mean curvature} on edge $e_{ij}$ (Fig.~\ref{fig: face_edge-constraint_net}(c)) is defined as $H_{ij}\!\triangleq\! \vert e_{ij}\vert \tan(\theta_{ij}/2)$. The integrated mean curvature on \textit{face} $i$ is the sum of its integrated edge curvatures: $H_i\!\triangleq\! \sum_{j\in\calN(i)} H_{ij}$. The discrete mean curvature $h_i\!\triangleq\! H_i/A_i$ follows by turning $H_i$ into a density over the face. 
With this, the discrete mean curvature satisfies a discrete counterpart to Steiner's formula. Steiner's formula is a characterization of mean curvature that relates it to the relative change of area when offsetting the surface in the normal direction $n$ by a distance $t$ (replace $A_i$ by an infinitesimal area element $dA$ in Eq.~\eqref{eq: steiner} for the original formula):
\begin{equation}
A_i^{(t)} = A_i(1+h_i t + \calO(t^2)) \,.
\label{eq: steiner}
\end{equation}

\paragraph{Geometry in the quaternions.} Quaternions $\bbH$ provide a natural algebraic language for geometry in $\R^3$, much like complex numbers 
for planar geometry. Let $\{1,\bbi,\bbj,\bbk\}$ denote a basis for $\bbH$. Elements $v\!\triangleq\!(v_x,v_y,v_z)\!\in\!\R^3$ are identified with pure imaginary quaternions $v_x \bbi + v_y \bbj + v_z \bbk\in \mathrm{Im}~\bbH\triangleq \mathrm{span}\{\bbi,\bbj,\bbk\}$, so that surfaces are naturally immersed in $\mathrm{Im}~\bbH$. Denote by $\bar{q}\triangleq a-(b\bbi\!+\!c\bbj\!+\!d\bbk)$ the quaternionic conjugate of $q\!\triangleq\! a\!+\!b\bbi\!+\!c\bbj\!+\!d\bbk\in\bbH$. The norm $\vert q\vert$ of $q$ is defined as the square root of $\bar{q}q=a^2\!+\!b^2\!+\!c^2\!+\!d^2$. All $q\neq 0$ admit an inverse $q^{-1}=\bar{q}/\vert q\vert^2$. Again like complex numbers, quaternions admit a polar decomposition $q\!=\!se^{\theta u}\!=\!s(\cos(\theta)+\sin(\theta)u)$ with $u\!\in\! \mathrm{Im}~\bbH$ a unit vector, which makes their geometric meaning more explicit. Indeed, $v\mapsto qvq^{-1}$, also known as conjugation by $q$, expresses rotation around $u$ by an angle $2\theta$. In the same vein, the expression $\tilde{v}=qv\bar{q}$ conveniently expresses a similarity transformation: $\tilde{v}$ corresponds to $v$ rotated around $u$ by $2\theta$ and rescaled by $s^2$. 

\paragraph{Hyperedges.} Every edge in the net $\calX$ is associated with a quaternion $E_{ij}$ dubbed \textit{hyperedge}, with real part equal to the integrated mean curvature $H_{ij}$ at the edge, and imaginary part equal to the embedding $e_{ij}\in\mathrm{Im}~\bbH$ of the edge:
\begin{equation}
E_{ij}\triangleq H_{ij} + e_{ij}\in \bbH \,.
\label{eq: hyperedges}
\end{equation}
Hyperedges are the fundamental structure on which discrete spin transformations act. 
They summarize all the geometric information that, along with the mesh combinatorics, allows to reconstruct the discrete surface immersion (Appendix~\ref{sec: hyperedges}). With that, spin transformations are introduced in a straightforward manner.

\section{Discrete Spin Transformations}
\label{sec: spin transformations}

\paragraph{Discrete spin transformations.} A discrete spin transformation $\phi$ associates a single quaternion $\phi_i$ with each face $i$ of a face edge-constraint net. The transformation acts on hyperedges $E_{ij}$ and face normals $n_i$ as follows:
\begin{equation}
\label{eq: discrete spin transformation}
\begin{aligned}
  E_{ij} &\mapsto \tilde{E}_{ij}=\bar{\phi}_i E_{ij} \phi_j\, ,\\
  n_i &\mapsto \tilde{n}_i=\phi_i^{-1} n_i \phi_i \, . 
\end{aligned}
\end{equation}
The elegance of the construct lies in the fact that Eq.~\eqref{eq: discrete spin transformation} does transform a face edge-constraint net into another edge-constraint net. This is easily checked (cf.~\cite{ye2018unified}), with the main elements of the proof stemming from the geometric interpretation of hyperedges (Appendix~\ref{sec: hyperedges}) and from the constraint on face normals. Furthermore, discrete spin transformations $E\rightarrow_{\phi} \tilde{E}$ are trivially invertible: $\tilde{E}\rightarrow_{\phi^{-1}}\!E$. The \textit{integrability} condition that each face in the new net closes ($\sum_j\!\tilde{E}_{ij}\in\!\R$) is equivalent to the existence of a real valued function $\rho: i\mapsto \rho_i\in\R$ over faces such that:
\begin{equation}
\label{eq: closedness}
D_{\calX} \phi = \rho A \phi \,.
\end{equation}

Equation~\eqref{eq: closedness} is the cornerstone of the framework. $D_{\calX}$ is henceforth referred to as the \textit{intrinsic} Dirac operator. $D_{\calX}$ sends a quaternionic function over faces to another one such that $(D_{\calX}\phi)_i\triangleq \sum_j E_{ij}\phi_j$. Left multiplying both sides by $\bar{\phi}_i$, the closedness constraint on faces is immediately apparent: $\bar{\phi}_i (D_\calX\phi)_i = \sum_j \tilde{E}_{ij}$ must be real-valued. For ease of exposition, the expression in the introduction is formulated using a slightly different yet immediately related operator, the \textit{extrinsic} Dirac operator $(D_e\phi)_i\!\triangleq\! \sum_j E_{ij}(\phi_j-\phi_i)=(D_\calX \phi)_i - H_i \phi_i$. It also discards the normalization by $A$ as in~\cite{ye2018unified}. 
The proposed normalization however mirrors more faithfully the smooth counterpart of the present setting (see \eg~\cite{kamberov1996bonnet}).

The intrinsic Dirac operator $D_{\calX}$ 
creates an explicit relationship between a spin transformation $\phi$ and the discrete (resp. integrated) mean curvature $\tilde{h}$ (resp. $\tilde{H}$) of the new net, namely $\bar{\phi}_i (D_\calX\phi)_i = \tilde{H}_i\triangleq \tilde{h}_i \tilde{A}_i$ as long as the new net closes. Coupling with Eq.~\eqref{eq: closedness}, 
\begin{equation}
\label{eq: rho interpretation}
\tilde{h}_i \tilde{A}_i=\rho_i A_i \vert\phi_i\vert^2 \, .
\end{equation}
When $\phi\!\coloneqq\! 1$ is the identity transform, $\rho_i\eq h_i\eq\tilde{h}_i$. For smooth $\vert\phi_i\vert$ and from Eq.~\eqref{eq: discrete spin transformation}, $\rho_i\sqrt{A_i}\!\approx\!\tilde{h}_i \sqrt{\tilde{A}_i}$. In other words, $\rho_i$ jointly describes the mean curvature and length element. This quantity is precisely known as the mean curvature half-density $h\vert df\vert$ in the smooth setting, and is generally in one-to-one correspondence with a given shape. Finally, with the \textit{extrinsic} Dirac operator, the corresponding $\mu$ describes a \textit{change} in half-density instead: $\tilde{h}_i \tilde{A}_i= (h_i+\mu_i) A_i \vert\phi_i\vert^2$.

\begin{figure}[t]
\includegraphics[width=\textwidth]{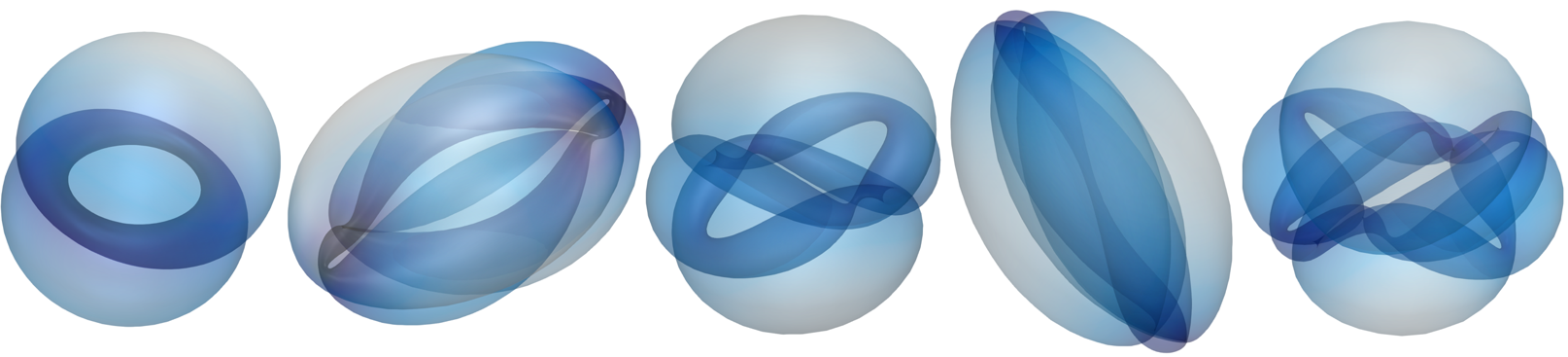}
\caption{A few leading eigenvectors of the intrinsic Dirac operator for the unit sphere, visualized as surface immersions (color map: eigenvector magnitude).} \label{fig: Dirac immersions}
\end{figure}

\paragraph{Dirac operators.} Dirac operators $D_\calX$ and $D_e$ have a number of properties that make them appealing for various tasks in shape analysis. $D_\calX$ and $D_e$ are self-adjoint operators. Dirac operators relate to square roots of the Laplace--Beltrami operator $L$. Whereas $L$ captures the intrinsic manifold geometry and is invariant by isometry, the Dirac operators can disambiguate much more about the surface \textit{immersion} into $\R^3$. We refer the reader to~\cite{liu2017dirac,ye2018unified} for a discussion from this perspective. The eigenvectors of Dirac operators all satisfy Eq.~\eqref{eq: closedness} and thus provide new immersions of the abstract manifold into $\R^3$ (new transformed $\tilde{\calX}$). The first (null) eigenvector of $D_e$ is trivial. $D_\calX$ cannot have a null eigenvalue for closed surfaces (\eg spherical topology) of practical interest in the present work, since that would result in a minimal closed surface with everywhere zero mean curvature. The smallest eigenvector of $D_\calX$ provides a generally non-trivial immersion with higher smoothness than the original shape (lower Willmore energy $\int \vert h\vert^2 dA$). Ye et al.~\cite{ye2018unified} explore this mechanism for the purpose of surface fairing. The next leading eigenvectors give some geometric insight into $D_\calX$ (Fig.~\ref{fig: Dirac immersions}). In this work however, we investigate a strategy closely related to~\cite{crane2013robust} with a fine-grained control over the surface deformations.

\section{Algorithms}
\label{sec: algorithms}

\begin{rmk}Quaternions $q$ admit representations $M[q]$ as $4\!\times\! 4$ real matrices (Eq.~\eqref{eq: matrix representation}), so that standard linear algebra libraries can be used to solve quaternionic linear systems. In particular, $M[\bar{q}]=M[q]^\T$, thus Hermitian (quaternionic) forms are represented by real symmetric matrices. We denote real vectors and matrix representations below with upright bold symbols.
\begin{equation}
\label{eq: matrix representation}
\begingroup
\setlength\arraycolsep{3pt}
M[q]\triangleq\begin{bmatrix*}[r]
   a & -b & -c & -d \\
   b &  a & -d &  c \\
   c &  d &  a & -b \\
   d & -c &  b &  a
   \end{bmatrix*}
\endgroup \,.
\end{equation}
\end{rmk}

\subsubsection{Overview.} The scalar function $\rho$ introduced in section~\ref{sec: spin transformations} provides the primary degrees of freedom for mesh manipulation, and it tightly relates to mean curvature. 
Of course only a subset of functions $\rho$ can be associated with some $\phi$ such that the integrability condition Eq.~\eqref{eq: closedness} is satisfied. Namely, $D_\rho\!\triangleq\!D\!-\!\rho$ should have a null eigenvalue. This leads Crane et al.~\cite{crane2011spin} to solve for the smallest eigenvalue $\gamma$ and eigenvector $\phi$, yielding a solution of Eq.~\eqref{eq: closedness} up to a small constant shift: $D\phi\eq(\rho\!+\!\gamma)\phi$. We propose instead to formulate the objective $D_{\!\rho}\phi\!\simeq\!0$ as a minimization problem. This gives fine-grained control to add specifications (\eg smoothness, area distortion), 
many of which can be efficiently expressed as linear(ized) constraints or quadratic regularizers, within a unified formulation. 
Thus finding $\phi$ amounts to solving a quadratic problem:
\begin{equation}
\label{eq: QP}
\argmin_\upphi \, \underbrace{\upphi^\T(\brmD_\mrho \brmA^{\!-1} \brmD_\mrho)\upphi}_{D_{\!\rho}\phi\simeq 0} + \underbrace{(\upphi\!-\!1)^\T\alpha \brmR(\upphi\!-\!1)}_{\text{regularization}}\, ,
\end{equation}
under a set of linear constraints on $\upphi$. $\brmA$ is a diagonal matrix of face areas. In practice we set $\brmR$ to $\brmA\!+\!\beta \brmL_f$, where $\brmL_f$ is an integrated Laplacian over faces. The eigensystem actually solved in~\cite{crane2011spin} closely relates to the simplest case where there are no constraints and $\beta\!\coloneqq\!0$.

Overall, the procedure is as follows: prescribe a scalar function $\rho$ for a target shape or curvature change (sec.~\ref{sec: applications}); then solve for the spin transformation $\phi$ (Eq.~\eqref{eq: QP}); finally compute new hyperedges (Eq.~\eqref{eq: discrete spin transformation}) and solve a linear system for the new vertex coordinates (Eq.~\eqref{eq: edge integration}). The steps are typically iterated over, resulting in a flow.

\subsubsection{Computing the new immersion.} Let transformed edges $\tilde{e}_{ij}\eq \mathrm{Im}\,\tilde{E}_{ij}$ be indexed by their start and end vertices $v\rightarrow v'$. Vertex coordinates $\tilde{f}\colon v\!\in\! V\mapsto \tilde{f}_v$ satisfy $\tilde{f}_{v'}-\tilde{f}_{v}=\tilde{e}_{v\rightarrow v'}$. In practice, we solve the mathematically equivalent (Appendix~\ref{sec: edge integration}) linear system 
\begin{equation}
\label{eq: edge integration}
\Updelta \tilde{f}\eq \nabla\cdot\tilde{e}\, ,
\end{equation} where $\Updelta$ and $\nabla\cdot$ are the standard discrete (cotangent) mesh Laplacian and divergence operators~\cite{meyer2003discrete}. This method of integration is robust to numerical errors. The Laplacian and divergence are computed w.r.t. either the source ($e$) or target ($\tilde{e}$) mesh metric (with empirically identical results). A benefit of working from a discrete setting is that no discretization error is introduced from $\phi$ to the corresponding $\tilde{f}$. 

\subsubsection{Geometrically constrained flows.} The proposed formulation (Eq.~\eqref{eq: QP}) enables fine-grained control over the flow by prescribing additional constraints. For instance, the method extends to topologies beyond spherical by adding an \textit{exactness} constraint (Appendix~\ref{sec: integrability constraints}). The mapping can also be encouraged to preserve angles (\ie conformality) and minimize area distortion. Conformality is key in preserving mesh quality across exceptionally large deformations, which prevents considerable loss of numerical stability. It is intuitively described as circles being locally transformed into circles, or indeed texture-preserving (Fig.~\ref{fig: leading figure}). Quasi-conformality is inherent to the present framework. From Eq.~\eqref{eq: discrete spin transformation}, the relative length of edges is preserved as soon as $\vert \phi_i\vert$ varies smoothly across faces. On the other hand large area distortion can be introduced, particularly in regions of high curvature. In some applications, we may prefer to trade off some distortion of angles for a better preservation of areas. We note again from Eq.~\eqref{eq: discrete spin transformation} that the magnitude $\vert\phi_i\vert^4$ of the spin transformation relates to the local change of area. Thus scale changes $\log{\tilde{A}_i/A_i}$ (up to global rescaling) can be penalized via a linearized soft constraint over $\phi$ (Appendix~\ref{sec: area distortion penalty}). 

\subsubsection{Filtering in curvature space.} As described in~\cite{crane2013robust} in a related setting, the flow of the spin transformation can also be altered by directly manipulating $\rho$. The rate of change for geometric features of various scales can be tweaked by manipulating its frequency spectrum. Moreover some constraints can be efficiently enforced by orthogonal projection of $\rho$ onto a linear subspace. In particular, Appendix~\ref{sec: integrability constraints} derives alternative integrability conditions in the form of simple linear constraints on $\rho$, for the proposed discrete geometric framework.

\section{Applications}
\label{sec: applications}

This section showcases the approach on a collection of structured meshes of subcortical structures from the UK Biobank database \cite{miller2016multimodal}. The typical mesh size is of a few thousand nodes (up to $20$k). The framework was implemented in numpy. The tool mostly relies on efficient (sparse) linear algebra. Experiments were run on a standard laptop (i7-8550U CPU @ 1.80GHz).

\begin{figure}[t]
\includegraphics[width=\textwidth]{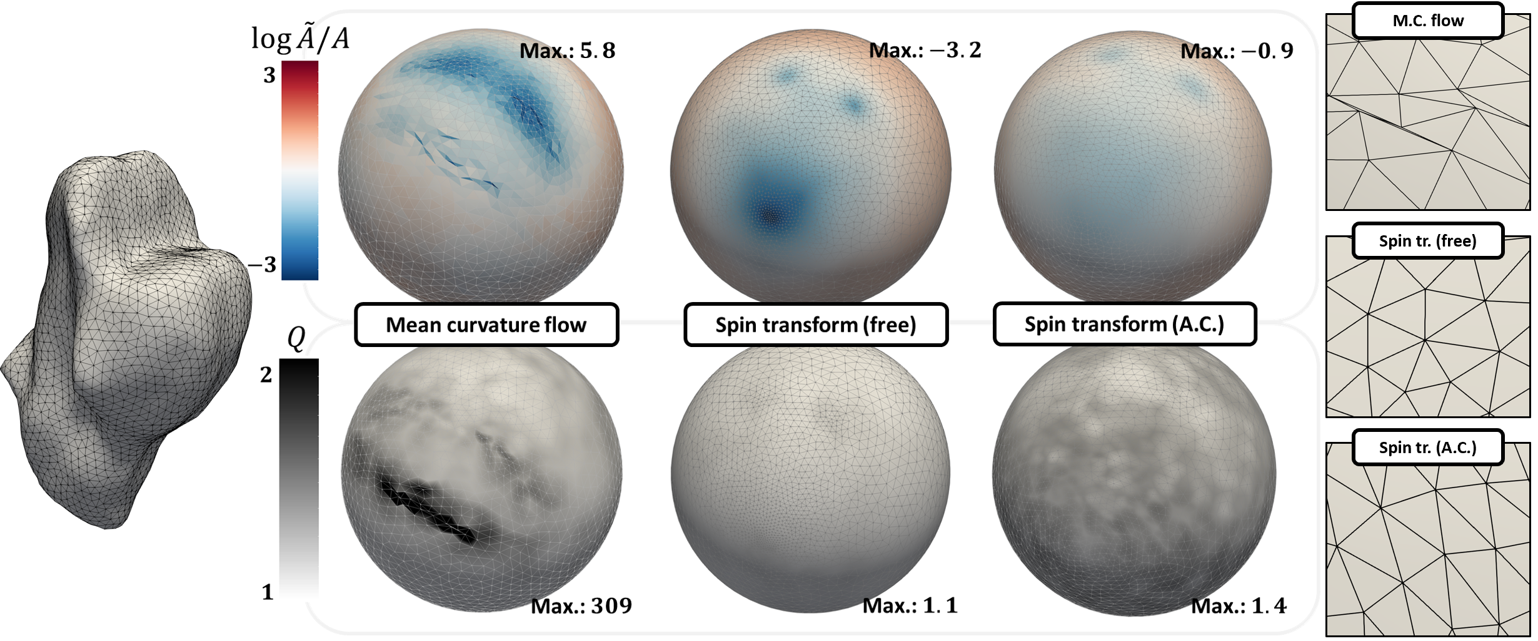}
\caption{Example surface flow of a subcortical structure (brain stem) to the reference sphere. Comparison of discrete spin transformations with an incompressible mean curvature (MC) flow. (Left) The brain stem. (Middle) Area distortion (top row, $0$ distortion is best) and conformality error (bottom row, $Q\eq1$ is best) displayed over the reference geometry. (Right) Zoom on the flowed triangulated mesh (A.C. $\equiv$ area constraint flow; free $\equiv$ unconstrained). Unlike MC flows, spin transformations naturally preserve the triangulation quality and are numerically stable. The area constrained variant yields a reasonable trade-off between preserving angles and areas without introducing unexpected artefacts.} \label{fig: mean curvature vs. spin}
\end{figure}

\subsubsection{Surface Fairing.} Surface fairing is the process of producing successively smoother approximations of a mesh geometry $f$. Most algorithms proceed by minimizing a fairing energy, such as the membrane energy $E_M(f)\!\triangleq\! \int_S \vert \nabla f \vert^2 dA$ or the Willmore functional $E_W(f)\!\triangleq\! \int_S h^2 dA$. Recalling that $\Delta f = h \vec{n}$ and ignoring the dependence of $\Delta$ on $f$, gradient descent on $E_M$ (resp.~$E_W$) yields $\dot{f}\!\propto\! \Delta f$ ($\dot{f}\!\propto\! \Delta^2\! f$). The former yields the widespread \textit{mean curvature flow} $\dot{f}\!\propto\! -h \vec{n}$ that iteratively evolves points along the surface normal $\vec{n}$ with a magnitude proportional to the mean curvature $h$. Crane et al.~\cite{crane2013robust} first suggested in the context of spin transformations to optimize $E_W$ directly w.r.t. $h$, yielding the simple flow $\dot{h}\!\coloneqq\! h$ \textit{in curvature space}. A benefit of the approach is to decouple time and spatial integration, yielding numerically stable solutions across large time steps. We follow the same strategy. The prescribed change of curvature $\delta h\!\coloneqq\! -\tau h$ is then (optionally filtered and) integrated into a new surface immersion $\tilde{f}$, by computing the corresponding spin transformation as per section~\ref{sec: algorithms}. Specifically, for a given target curvature $\bar{h}_i$ (say $h_i\!+\!\delta h_i$) and area $\bar{A}_i$, we let $\rho_i\!\coloneqq\! \bar{h}_i\sqrt{\bar{A}_i/A_i}$ (section~\ref{sec: spin transformations}). The standard unconstrained optimization (Eq.~\eqref{eq: QP}) regularized with the face Laplacian $\brmL_f$ (or one of its powers) yields quasi-conformal transformations (Fig.~\ref{fig: leading figure}). Large steps $\tau\eq0.5$--$1$ typically remain stable. Whether $\phi$ is numerically integrable can be checked by monitoring the discrepancy between edges $\tilde{E}$ integrated as per Eq.~\eqref{eq: discrete spin transformation}, and edges recomputed from $\tilde{f}$ (\textit{after} getting $\tilde{f}$ from Eq.~\eqref{eq: edge integration}). The closedness generally holds within a few percent across several large steps without an explicit constraint, and within $10^{-6}$ with an explicit constraint (Appendix~\ref{sec: edge integration}). A trade-off between conformality and area distortion is achieved by weighing in a soft constraint on the square norm of the logarithmic area distortion (Fig.~\ref{fig: mean curvature vs. spin}).

\subsubsection{Comparison to Mean Curvature Flow.} The procedure is compared with an incompressible mean curvature flow. Incompressibility is enforced to make the flow \textit{less} prone to develop singularities, by adding a balloon energy $\langle h\rangle \vec{n}$, where $\langle h\rangle\!\eq\! \int_S hdA$ is the average mean curvature. Two metrics of interest, defined over the mesh surface, are the conformality error $Q$ and the logarithmic area distortion $\epsilon_s=\log{\tilde{A}/A}$ (after normalising to the same total area). The quality factor $Q$ measures how close-to-conformal a transformation is, as the ratio of the largest to smallest eigenvalues of the Jacobian of the mapping from $f$ to $\tilde{f}$. For a conformal deformation, $Q$ is identically $1$ throughout the mesh. However the area distortion $\epsilon_s$ may become significant. Fig.~\ref{fig: mean curvature vs. spin} exemplifies the general observation that the mean curvature flow realises a suboptimal trade-off between angle and area preservation. As expected, unconstrained discrete spin transformations are quasi-conformal. Unavoidable area distortion is introduced but mesh elements retain their original quality (right column, top and middle). To contrast, the mean curvature flow arbitrarily destroys the mesh quality, angle and area ratios in regions of high curvature. Area constrained discrete spin transformations implement a sensible compromise, whereby~(i) area distortion is lessened;~(ii) numerical stability is preserved;~(iii) the conformal error increases rather uniformly over the entire mesh, leading to a graceful, slower loss of mesh quality. For surface fairing to a sphere, averaged over a random subset of $100$ meshes in the dataset and taking the \textit{maximum} over the mesh surface, we get the following -- mean curvature flow: $Q\eq97\pm165$, $\epsilon_s\eq 4.1\pm 1.5$;~unconstrained spin transformation: $Q\eq 1.42\pm0.08$, $\epsilon_s\eq2.9\pm 0.3$;~area constrained: $Q\eq 1.7\pm0.2$, $\epsilon_s\eq0.85\pm 0.05$. For the area constrained spin transform, the maximum area discrepancy simply reflects a user-specified soft target.

\subsubsection{Mesh Extrusion.} The task is now to reconstruct (``extrude'') a shape of interest back from a reference sphere, given its mean curvature $h^{\star}$ and area $A^{\star}$ mapped onto the sphere surface. There is to our knowledge very little done in that direction, even in related work~\cite{crane2013robust,ye2018unified}. To emphasize, we only wish to recover the original mesh up to \textit{pose} and scale. Encoding scale presents little difficulty, and shape is invariant under changes of pose. To evaluate the reconstruction accuracy, we rigidly align and rescale the extruded shape to the original one, and compute the maximum distance from points on the extruded mesh to the original surface. The strategy for extrusion closely mirrors that of mesh fairing, whereby we get $\bar{h}_i$ from $\delta h_i\!\coloneqq\! h_i^{\star}\!-\!h_i$, and set $\rho_i\!\coloneqq\! \bar{h}_i\sqrt{A_i^{\star}/A_i}$. As a preliminary comment, note that the degree of challenge regarding mesh extrusion critically depends on the exact experimental setting and goal, as contrasted in the two following settings. The first experiment aims to estimate the accuracy that can be reached in the best scenario (somewhat upper bounded by the registration error). We take a collection of $300$ subcortical meshes from the UK Biobank (incl. brain stems, caudate, putamen, accumbens, amygdala, hippocampus, thalamus, palladium) and flow them onto the unit sphere. We do not perform remeshing, only interpolating relevant maps to nodes and back to faces. We then directly reconstruct the mesh as described above. On average over the dataset, the maximum point-to-surface error is of $0.4$mm. The distribution of error is widely spread over different structures, the most challenging being caudates ($1.4$) and hippocampi ($1.2$); and the least ones being the accumbens, amygdala, palladium and thalamus ($\sim0.01$--$0.02$). This matches our expectations, given that caudates and hippocampi are in fact highly non spherical. Thus very significant area or angle distortion is introduced when mapping onto the sphere. The second experiment investigates a more challenging setup, whereby the flowed surface is remapped onto a reference sphere with uniform meshing. Shape-specific vertex density as well as face aspect ratio, which reflect the area and angle distortion introduced during the fairing, are thus discarded. We experiment with a set of $100$ brain stems (Fig.~\ref{fig: brain stem reconstruction}), which represent a happy medium between the most challenging and trivial structures, with a maximum reconstruction error of $1.4\pm0.3$mm ($2$--$4$\%). For the most challenging structures, various strategies to guide the reconstruction using either additional information obtained during the flow, or multiscale approaches with hierarchical encoding could be considered. This is left to explore in future work.

\begin{figure}[t]
\includegraphics[width=\textwidth]{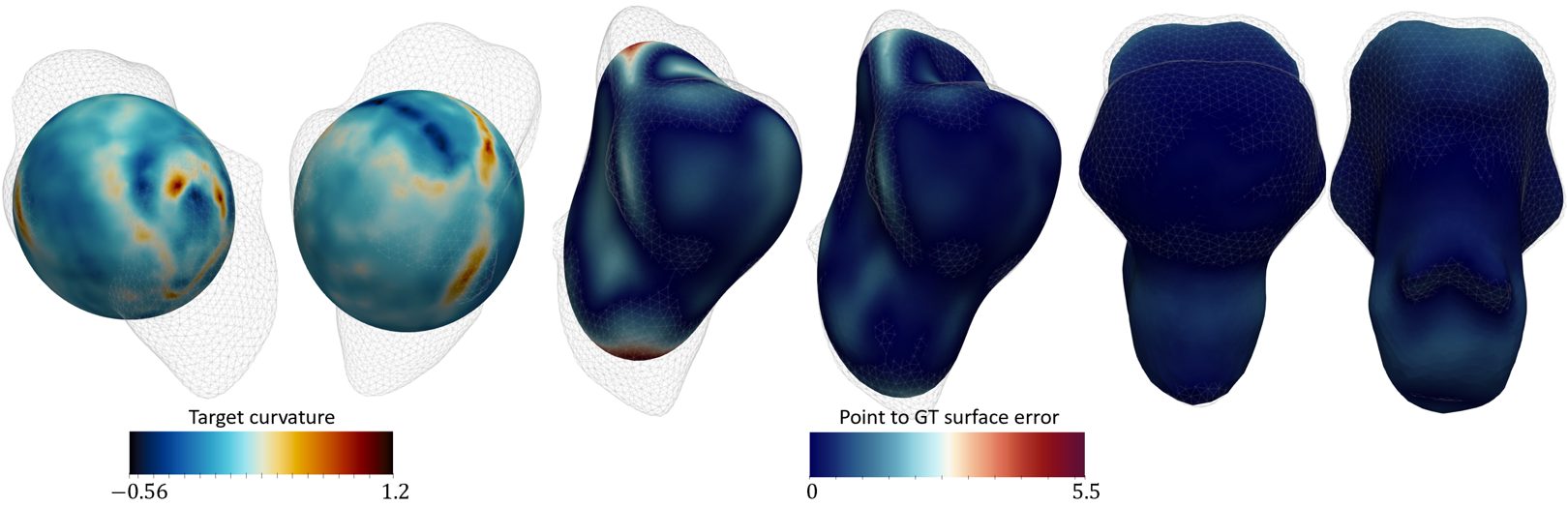}
\caption{Example extrusion of a brain stem from the reference sphere. The original shape is overlaid as a wireframe. ($1$st and $2$nd) Close to the initial stage. The target mean curvature map is displayed, rather than the reconstruction error. Note that the shape flow in the next stages intuitively matches the information captured in these maps. ($3$rd and $4$th) Intermediate stages in the flow, with overlayed reconstruction error. ($5$th to last) Reconstructed mesh from two views.} \label{fig: brain stem reconstruction}
\end{figure}

\section{Conclusion}

We have presented a method to manipulate surface meshes across very large deformations by prescribing mean curvature (half-density). The framework is
well suited for mesh fairing and extrusion, \eg to map shapes to, or back from a unit sphere. As a perspective, we believe the approach to have potential for
pose-invariant shape analysis, specifically for \textit{generative} modeling. Indeed mean curvature together with the metric generally is in one-to-one correspondence with the (closed) shape; this is in particular true for a spherical topology. We have shown how spin transformations computationally implement this insight. Therefore the shape \textit{geometry} could be losslessly encoded as a scalar \textit{function} on a template, making the modeling task more amenable to \textit{learning}. In the smooth setting, spin transformations are a subgroup of conformal maps. This partly explains their numerical stability across large flow steps, a property inherited in the discrete setting. However, conformal maps can introduce significant area distortion, \eg when flowing highly curved objects. An advantage of \textit{discrete} spin transformations is to relax exact conformality, and allow the user to trade off angle for area preservation. 

\section*{Acknowledgments}\label{sec:Acknowledgments}
This work is supported by the EPSRC (grant ref no.
EP/P023509/1) and the European Research Council (ERC) under the European Union's Horizon 2020 research and innovation programme (grant agreement No 757173, project MIRA, ERC-2017-STG). DC is also supported by CAPES, Ministry of Education, Brazil (BEX 1500/15-05). KK is supported by the President's PhD Scholarship of Imperial College London. IW is supported by the Natural Environment Research Council (NERC).

%
%
%
 \bibliographystyle{splncs04}
 \bibliography{paper}

\newpage

\appendix
\section{Geometric interpretation of hyperedges}
\label{sec: hyperedges}

\begin{remark}
Letting $u_{ij}\!\triangleq\!e_{ij}/\vert e_{ij}\vert$ and after straightforward manipulations, we get: 
\begin{equation}
E_{ij}=\frac{\vert e_{ij}\vert}{\cos{(\theta_{ij}/2)}}\, \exp\left(\frac{\pi-\theta_{ij}}{2}u_{ij}\right)\, .
\end{equation} Thus conjugation by $E_{ij}^{-1}$ sends any vector lying on face $i$ to face $j$, and $-n_i$ to $n_j$. Intuitively speaking, $E_{ij}$ carries information about a connection structure between the affine spaces of faces $j$ and $i$~\cite{hoffmann2018discrete}.
\end{remark}

\begin{remark}
$\bar{E}_{ij}=H_{ij}-e_{ij}=H_{ji}+e_{ji}=E_{ji}$. Moreover, assuming the mesh to be closed, edges sum to $0$ over any given face, so that $\sum_j E_{ij} = H_i\in\R$.
\end{remark}

\begin{remark} 
Definition of hyperedges as per Eq.~\eqref{eq: hyperedges} may seem somewhat arbitrary. In fact, $E_{ij}$ is necessarily of the form $\vert e_{ij}\vert \tan{\alpha_{ij}} + e_{ij}$, up to a multiplicative constant, under the following mild conditions: (a) the imaginary part of $E_{ij}$ is along $e_{ij}$; (b) $\bar{E}_{ij}=E_{ji}$; and (c) $\sum_j \!E_{ij}\!\in\!\R$ iff face $i$ closes. Relating $\alpha_{ij}$ to the bending angle $\theta_{ij}$ is sufficient to guarantee that spin transformations transform a net into another valid net (\ie the real part $\tilde{H}_{ij}$ of the transformed edge $\tilde{E}_{ij}$ is provably consistent with the constructive definition above).
\end{remark}

\section{Edge Integration}
\label{sec: edge integration}

Let $\nabla\tilde{f}\colon \epsilon\!\triangleq\!(v\!\rightarrow\! v') \mapsto (\tilde{f}_{v'}\!-\!\tilde{f}_{v})\!\triangleq\! \nabla\tilde{f}(\epsilon)$ the discrete gradient. We are looking for $\tilde{f}$ s.t. $\nabla\tilde{f}(\epsilon)=\tilde{e}_{\epsilon}$. When such an $\tilde{f}$ exists, $\tilde{e}$ is said to be \textit{exact} (as a discrete $1$-form, defined over edges $\epsilon$). In that case, Eq.~\eqref{eq: edge integration} follows by taking $\nabla\cdot$ on both sides.

Define the discrete curl operator $[\nabla\!\times\! \tilde{e}](i)\!\triangleq\! \sum_{j\in\calN(i)}\tilde{e}_{ij}$, indexing as in section~\ref{sec: geometric setting}. The curl sends a 1-form (over edges) to a 2-form (over faces). If $[\nabla\!\times\! \tilde{e}]$ vanishes everywhere, $\tilde{e}$ is said to be \textit{closed}. It is easy to verify that $\nabla\!\times\!\nabla\tilde{f}$ is everywhere zero, so that \textit{exactness} always implies \textit{closedness}. For a (discrete) spherical topology, the converse holds: closedness implies exactness. Now let $\tilde{e}\eq\mathrm{Im}\, \tilde{E}$ be generated by a spin transformation $\phi$ acting on hyperedges $E$ with $e$ closed. Then $\tilde{e}$ is closed iff Eq.~\eqref{eq: closedness} is satisfied (immediate from the definition of $D_\calX$, cf. section~\ref{sec: spin transformations}).

\begin{remark}
High-level elements of constructive proof are derived from the mesh being simply connected. It is path connected, so we can fix a vertex $v$ and reach any vertex $v'$ from $v$ by following a path $\gamma({v\!\rightarrow\! v'})$ on edges. Let $\tilde{f}_{v'}\!\coloneqq\!\tilde{f}_v+\sum_{\epsilon\in\gamma(v\!\rightarrow\! v')}\tilde{e}_{\epsilon}$ obtained by summing edges along the path. $\tilde{f}$ is well defined because the value at $v'$ is independent of the path. Indeed let $\gamma_1$, $\gamma_2$ two paths from $v$ to $v'$. Following $\gamma_1$ then the reverse of $\gamma_2$, we run a closed loop. Self-intersections are removed without loss of generality. One can prove by induction on the loop length that edges sum to $0$ over the loop, thus the sum over $\gamma_1$ and $\gamma_2$ are equal. This holds for vertices on a single face by \textit{closedness}. Closed loops of arbitrary length can always be incrementally shrunk down to this case (by \textit{simple} connectivity), without changing the sum of edge values (by \textit{closedness}). 
\end{remark}

\subsubsection{Non-simply connected topologies.} Consider a path connected mesh, but possibly with handles (note that closed loops circling a handle cannot be continuously shrunk down to a trivial loop). A \textit{closed} 1-form $\tilde{e}$ can fail to be \textit{exact} if it has a non-zero \textit{harmonic} component, \ie if it can be written as $\tilde{e}\eq \omega+\nabla\tilde{f}$ for some discrete 0-form $\tilde{f}\colon V\!\rightarrow\!\mathrm{Im}\, \bbH$ and harmonic 1-form $\omega\colon E\!\rightarrow\! \mathrm{Im}\, \bbH$ (s.t. $\omega\!\neq\! 0$ is closed and $\Updelta_1 \omega\eq 0$). Equivalently $\omega$ is closed with vanishing divergence $\nabla\cdot \omega\eq 0$. While Eq.~\eqref{eq: edge integration} still admits a solution $\tilde{f}$, the corresponding edges, $\tilde{e}\!-\!\omega$, are not the prescribed ones. Fortunately, convenient exactness constraints can be derived via the following theorem.

\begin{theorem}[Helmholtz--Hodge decomposition]
The $L^2$ space of (alternating) 1-forms (edge flows) $L_{\wedge}^2(E)$ on $\calG\!\triangleq\!(\calV,\calF,\calE)$ admits an orthogonal decomposition into subspaces of co-exact, harmonic, and exact forms:
\begin{equation}
L_{\wedge}^2(E)=\mathrm{im}([\nabla\!\times]^\T) \oplus \underbrace{\mathrm{ker}(\Updelta_1) \oplus\mathrm{im}(\nabla)}_{\mathrm{ker}(\nabla\!\times)}\, ,
\end{equation}
where $\Updelta_1\!\triangleq\! -\nabla[\nabla\cdot] + [\nabla\!\times]^\T[\nabla\!\times]$ is the so-called \textit{graph Helmholtzian}. 
\end{theorem}

\begin{proof}
This is Hodge theorem in linear algebra, noting that $\nabla^\T=-[\nabla\cdot]$ and $\nabla\!\times\!\nabla\eq 0$. Moreover, by Hodge isomorphism theorem, the dimension of $\mathrm{ker}(\Updelta_1)$ is the $1$st Betty number $b_1$, \ie the number of handles for $\calG$ (generally small).
\end{proof} 

Concretely, let $\omega^1 \cdots \omega^{b_1}$ a set of null eigenvectors for the Helmholtzian $\Updelta_1$. Edges $\tilde{e}$ are integrable iff $\tilde{e}$ is closed (Eq.~\ref{eq: closedness}) and orthogonal to $\omega^k \nu$ ($k\eq 1\cdots b_1$, $\nu\eq \bbi,\bbj,\bbk$) w.r.t. the inner product on $L_{\wedge}^2(E)$: $\langle \tilde{e}\vert \omega^k \nu\rangle_{1,\bbH}\eq 0$. In Appendix~\ref{sec: integrability constraints}, these constraints on transformed edges are turned into constraints on the spin transformation $\phi$, or alternatively on the prescribed curvature map $\rho$.

\begin{remark}[Cotangent discretization]
\label{rmk: cotangent discretization}
Let $\calX$ a triangulated net. For $\R$-edge flows $g,\tilde{g}\in L_{\wedge}^2(E)$, define the inner product $\langle g\vert \tilde{g} \rangle_{1,\R}\!\triangleq\!\frac{1}{2}\sum_\epsilon w_\epsilon g(\epsilon)\tilde{g}(\epsilon)$. Set edge weights $w_\epsilon$ to the symmetric expression $w_{v\rightarrow v'}\!\triangleq\! \frac{1}{2}(\cot{\angle vv_1v'}+\cot{\angle v'v_2v})$ where $v_1$ (resp. $v_2$) complete the two triangular faces adjacent to the edge $v\!\rightarrow\! v'$. On $0$-form, define the standard inner product $\langle f \vert \tilde{f}\rangle_0\!\triangleq\! \sum_v f(v)\tilde{f}(v)/A_v$ where vertices are weighted by cell areas $A_v$. Define the divergence operator $\nabla\cdot\!\triangleq\! -\nabla^\T$ as the negative adjoint of the gradient $\nabla$. Then $\nabla\cdot$ is exactly the \textit{cotangent}-weighted divergence (the sum of outbound edge flows at $v$ weighted by $w_{v\rightarrow v'}$) and $\Updelta=\nabla\cdot \nabla$ the \textit{cotangent} Laplacian. With this, the relevant inner product $\langle\cdot\vert\cdot\rangle_1$ compatible with the \textit{cotangent} scheme is now specified.
\end{remark}

\section{Exactness Constraint}
\label{sec: integrability constraints}

From Appendix~\ref{sec: edge integration}, $E \!\rightarrow_\phi \!\tilde{E}$ is integrable if $\tilde{E}$ is closed and orthogonal to $\omega_k \nu$, where $\omega^k$ spans real-valued harmonic $1$-forms ($k\eq 1\cdots b_1$) and $\nu\eq \bbi,\bbj,\bbk$, w.r.t. the inner product on $L_{\wedge}^2(E)$, say $\langle\cdot\vert\cdot\rangle_{w_\epsilon}$ with the notations of Remark~\ref{rmk: cotangent discretization}.

\subsubsection{Closedness.} This is Eq.~\eqref{eq: closedness} and already core to the framework. Algorithmically, Eq.~\eqref{eq: QP} only guarantees the closedness to approximately hold, but we have observed very good agreement in practice without further constraint. If necessary, closedness can be strictly enforced as a set of $3\vert \calF\vert$ real-valued constraints ($3$ imaginary dimensions, $\vert \calF\vert$ faces), \eg by linearizing $\mathrm{Im}(\bar{\phi}_i D\phi_i)\eq 0$ around the current solution. 

\subsubsection{Exactness.} Exactness is guaranteed if $\mathrm{Re}(\nu \sum_\epsilon w_\epsilon\tilde{E}_\epsilon \omega_\epsilon^k)\eq 0$. $\omega_\epsilon$ and $\mathrm{Im}(\tilde{E})$ are alternating (\eg $w_{ij}=-w_{ji}$ with the conventions of section~\ref{sec: geometric setting}) so this rewrites as a set of $b_1$ ($3 b_1$ real-valued) constraints:
\begin{equation}
\label{eq: exactness quadratic constraint}
\mathrm{Im}~\sum_i \bar{\phi}_i \Big(\sum_{j\sim i} w_{ij} E_{ij}\omega_{ij}^k\phi_j\Big)=0\, ,
\end{equation}
that can be linearized around the current solution $\phi$. Alternatively, let us derive the corresponding constraint on $\rho$.  
Consider a time flow $\phi_t$, $\rho_t$ starting at $\phi_0=1$, with time derivative $\dot{\phi}$, $\dot{\rho}$ at $t\eq 0$. Deriving w.r.t. time, Eq.~\eqref{eq: closedness} becomes $D_e\dot{\phi}=\dot{\rho}$ and Eq.~\eqref{eq: exactness quadratic constraint} rewrites as:
\begin{equation}
2\sum_i \dot{\bar{\phi}}_i \Big(\sum_{j\sim i}w_{ij}E_{ij} \omega_{ij}^k\Big)= 2\sum_i A_i \dot{\bar{\phi}}_i v_i^k 
 \in \R\, ,
\end{equation}
where we use the alternating property to collapse the two sums and define $v^k\!\triangleq\! \nabla\!\times\!(wE\omega^k)$. In other words, for $\nu$ spanning $\mathrm{Im}\,\bbH$, $\big\langle \dot{{\phi}} \big\vert v^k\nu \big\rangle_{0,\bbH}\eq 0$. Let $z^k$ the unique solution to $D_ez^k\eq v^k$ and note that $D_e(z^k\nu)\eq (D_ez^k)\nu$. Finally, since $D_e$ is self adjoint, $\big\langle \dot{\phi}\big\vert D_ez^k\nu \big\rangle_{0,\bbH} = \big\langle D_e\dot{\phi}\big\vert z^k\nu \big\rangle_{0,\bbH}$ and we obtain:
\begin{equation}
\label{eq: exactness rho constraint}
\big\langle \delta\!\rho \big\vert z_{\nu}^k \big\rangle_{0,\R} =0, \quad k=1\cdots b_1, \quad\nu=\bbi,\bbj,\bbk\, .
\end{equation}
where $z_\nu^k$ are the three imaginary components of $z^k$. The constraint can be enforced by projection of the update $\delta\!\rho$ on the orthogonal subspace of the $z_\nu^k$. 

\subsubsection{To summarize:}
\begin{enumerate}[label=(\roman*)]
\item Compute the null eigenspace $\omega_1\cdots \omega_{b_1}$ of the Helmholtzian $\Updelta_1\!\triangleq\! -\nabla[\nabla\cdot] + [\nabla\!\times]^\T[\nabla\!\times]$
\item Compute $v^k\!\triangleq\! \nabla\!\times\!(wE\omega^k)$ and $z^k$ s.t. $D_ez^k\eq v^k$
\item Project $\delta\!\rho$ onto the orthogonal subspace of the imaginary components of $z^k$
\end{enumerate}

\section{Area distortion}
\label{sec: area distortion penalty}

\subsubsection{Overview.} We wish to penalize local scale changes $\log{A_i/A_i^0}$ (\ie $A_i$ moving away from the initial area distribution $A_i^0$), relative to the global rescaling $\sum_{i} A_i/\sum_{i} A_i^0$. Therefore the local scale change (with global rescaling factored out) writes as 
\begin{equation}
\label{eq: local scale change}
s_i \triangleq \log{A_i/A_i^0}-\log{\langle A_i\rangle / \langle A_i^0\rangle}, \,
\end{equation}
with $\langle A_i\rangle$ the average face area. We implement a soft constraint of the type $s_i^2\leq \epsilon^2$, where $\epsilon$ defines a tolerance for area distortion. $\epsilon$ can be set by the user or jointly adjusted over the course of the iterations. After introducing Lagrange multipliers $\lambda_i$, we are looking at penalties of the form $\sum_i \lambda_i s_i^2/2$, which we approximate by linearizing $s_i$ w.r.t. a variation $\delta\phi$ of the spin transformation $\phi$. We have found this mechanism to hold over large integration steps in practice. This can be better intuited by looking at the nature of the approximations made during the linearization (see below). The approximate quadratic energy is the sum of a sparse block diagonal matrix and a low-rank (dense) term. Woodbury matrix identities allow to solve quadratic systems involving this energy without directly storing or manipulating the dense matrix.

\subsubsection{Linearization of $s_i$.} We look for a linearized approximation of $\tilde{s}_i$ for a change $\delta \phi$ around the spin transformation $\phi$. We start by linearizing Eq.~\eqref{eq: discrete spin transformation}. Noting that $\phi_i+\delta \phi_i = \phi_i(1+\phi_i^{-1}\delta\phi_i)$, we get:
\begin{align}
\tilde{E}_{ij} & = \overline{(1+\phi_i^{-1}\delta\phi_i)} \cdot \overline{\phi}_i {E}_{ij}^0 {\phi}_j \cdot (1+\phi_j^{-1}\delta\phi_j)\\
	& = \overline{(1+\phi_i^{-1}\delta\phi_i)} {E}_{ij} (1+\phi_j^{-1}\delta\phi_j) 
\end{align}

Recalling from Appendix~\ref{sec: hyperedges} that $\vert E_{ij}\vert \cos{(\theta_{ij}/2)} = \vert e_{ij}\vert$ and taking the norm on both sides, we get:
\begin{equation}
\vert \tilde{e}_{ij}\vert \cos{(\theta_{ij}/2)} = \vert 1+\phi_i^{-1}\delta\phi_i \vert \cdot \vert {e}_{ij}\vert\cos{(\tilde{\theta}_{ij}/2)} \cdot \vert 1+\phi_j^{-1}\delta\phi_j\vert \, .
\end{equation}

We can ignore the change in the cosinus of the dihedral angle (to the first order) for simplicity. Secondly we assume $1+\phi^{-1}\delta\phi$ to be close to conformal. This assumption is coherent with the spirit of the framework, and should hold regardless if mesh quality is to be preserved locally in time. The overall transformation is still expected to progressively drift from quasi-conformality to accommodate area preservation. With this we can approximate the change in area for a small variation $\delta\phi$ from the change in edge length, and we get:
\begin{equation}
\tilde{A}_i\approx A_i \vert 1+ \phi_i^{-1}\delta\phi_i \vert^4 \, ,
\end{equation}
where $A_i$ is the area when applying $\phi$ to the initial face-edge constraint net, resp. $\tilde{A}_i$ when applying $\phi+\delta\phi$. This yields the following expression for $\tilde{s}_i$:

\begin{align}
\tilde{s}_i  & = s_i + \log{\vert 1+ \phi_i^{-1}\delta\phi_i \vert^4}-\log{\langle \vert 1+ \phi^{-1}\delta\phi \vert^4\rangle_A}\\
& =  s_i + \log{\frac{\vert \phi_i+ \delta\phi_i \vert^4}{\vert \phi_i \vert^4}}-\log{\left\langle \frac{\vert \phi+ \delta\phi \vert^4}{\vert \phi \vert^4}\right\rangle_A}\\
& \approx s_i + 4 \left(\frac{\langle \phi_i \vert \delta\phi_i \rangle_\bbH}{\vert \phi_i \vert^2} - \left\langle\frac{\langle \phi \vert \delta\phi \rangle_\bbH}{\vert \phi \vert^2} \right\rangle_A \right)
\end{align}
where $\langle \cdot \rangle_A$ denotes the spatial average weighted by the face areas $A$, whereas $\langle\cdot \vert \cdot\rangle_\bbH$ is the inner product on quaternions. In the last expression, we made use of $\vert \delta\phi_i\vert \ll \vert\phi_i\vert$, keeping only first order terms.

\subsubsection{Penalty matrix assembly.} Writing the penalty as $\frac{1}{2}\upphi^\T \brmQ \upphi - \rmF^\T \upphi$, the penalty matrix $\brmQ$ is the sum of a sparse block diagonal term and $3$ rank-1 terms, $4^2 (\text{diag}(Q_iQ_i^\T) + \rmL_1 \rmL_1^\T - \rmL_2\rmL_3^\T - \rmL_3 \rmL_2^\T)$. The derivations are tedious but straightforward, yielding:

\begin{equation}
\rmF_i = 4 \Big(\langle s\lambda\rangle_A \!-\! s_i \lambda_i\Big) \left\vert \frac{a_i}{\vert \phi_i\vert^2} \upphi_i \right \rangle_\bbH\, , \quad Q_i  = \left\vert \frac{\sqrt{a_i \lambda_i}}{\vert \phi_i\vert^2} \phi_i \right \rangle_\bbH \, , 
\end{equation}
\begin{equation}
\rmL_{j1} = \left\vert \frac{a_j\sqrt{\langle\lambda\rangle_A}}{\vert \phi_j\vert^2} \phi_j \right \rangle_\bbH \, , \quad
\rmL_{j2}  = \left\vert \frac{a_j\lambda_j}{\vert \phi_j\vert^2} \phi_j \right \rangle_\bbH \, , \quad
\rmL_{j3} = \left\vert \frac{a_j}{\vert \phi_j\vert^2} \phi_j \right \rangle_\bbH \, .
\end{equation}
with the use of bra-ket notation, and where $a$ stands for a normalised area $a\!\triangleq\! A/A_{\text{tot}}$. The sum of rank-$1$ updates might be degenerate (for instance it is rank-$1$ if all multipliers are equal). SVD decomposition can be used to derive an equivalent, non degenerate low-rank basis of vectors.

\end{document}